\documentclass[journal,onecolumn]{IEEEtran}
\usepackage{arydshln}
\usepackage{amsthm,amsmath,amssymb}
\usepackage{mathrsfs}
\usepackage{cite}
\usepackage{multirow}
\usepackage{tikz}
\usepackage{microtype}
\usepackage{hyperref}
\usepackage[capitalize]{cleveref}
\usepackage{verbatim}
\usepackage{xcolor}
\usetikzlibrary{matrix} 

\newtheorem{lemma}{Lemma}
\newtheorem{remark}{Remark}

\newtheorem{corollary}{Corollary}

\newtheorem{definition}{Definition}
\newtheorem{theorem}{Theorem}
\newcommand{\C}{\mathcal{C}}
\newcommand{\F}{\mathbb{F}}
\newcommand{\rank}{\mathrm{rank}}
\newcommand{\cF}{\mathcal{F}}
\newcommand{\cU}{\mathcal{U}}
\newcommand{\cE}{\mathcal{E}}

\newcommand{\cV}{\mathcal{V}}
\newcommand{\cP}{\mathcal{P}}

\newcommand{\bl}{\backslash}
\newcommand{\di}{\mathrm{diag}}
\newcommand{\cR}{\mathcal{R}}
\newcommand{\bO}{\mathbf{O}}
\newcommand{\co}{c}
\usepackage{etoolbox}
\let\bbordermatrix\bordermatrix
\patchcmd{\bbordermatrix}{8.75}{4.75}{}{}
\patchcmd{\bbordermatrix}{\left(}{\left[}{}{}
\patchcmd{\bbordermatrix}{\right)}{\right]}{}{}
\usepackage{url} 
\usepackage{hyperref} 

\begin{document}
\title{Lower Bounds on the Sub-Packetization of Optimal-Access MSR Codes for Multiple-Node Repair}
\author{Lewen Wang, \and Zihao Zhang, \and Sihuang Hu
		\thanks{
			The authors are with the Key Laboratory of Cryptologic Technology and Information Security, Ministry of Education, Shandong University, Qingdao, Shandong, 266237, China, 
            and School of Cyber Science and Technology, Shandong University, Qingdao, Shandong, 266237, China.
			S. Hu is also with Quan Cheng Laboratory, Jinan 250103, China.
		 Research partially funded by
		 National Key R\&D Program of China under Grant No. 2021YFA1001000, and
		 National Natural Science Foundation of China under Grant No. 12231014.
			Email: lewenwang3@gmail.com, zihaozhang@mail.sdu.edu.cn, husihuang@sdu.edu.cn
            (Corresponding author: Sihuang Hu.)
            }
            }
\date{}
\maketitle
\begin{abstract}
We establish lower bounds on the sub-packetization of optimal-access MSR codes in the context of multiple-node failures. 
These bounds generalize the tight bounds for single-node failure presented by Balaji et al. (IEEE Transactions on Information Theory, vol. 68, no. 10, 2022). 
Moreover, we utilize generating functions to provide a more refined analysis, further strengthening these bounds.
\end{abstract}

\section{Introduction}
Modern distributed storage systems (DSS) require both high reliability and low storage overhead, 
particularly in the face of frequent node failures. 
Erasure codes have become a fundamental tool in meeting these needs, compared to traditional replication techniques. 
Among these, Maximum Distance Separable (MDS) codes stand out, as they provide the highest reliability for a given storage overhead. 
A crucial challenge in DSS is the efficient repair of failed nodes. 

\subsection{Single-Node Repair}
Regenerating codes, introduced by Dimakis et al.~\cite{Dimakis2010}, 
were initially designed to address single-node repair and have demonstrated excellent performance in this context. 
Let \(\mathcal{C}\) be an \((n, k, \ell)\) array code over the finite field \(\mathbb{F}_q\). 
  A codeword \(C \in \mathcal{C}\) is represented as \(C = (C_1, \dots, C_n)^{T}\), 
where each codeword symbol \(C_i\) is a vector of the form \(C_i = (c_{i,1}, \dots, c_{i,\ell})\) over \(\F_q\) for \(i \in [n]\). 
The length $\ell$ of each codeword symbol is called the \emph{sub-packetization} of the code.
For an \((n, k, \ell)\) MDS array code, Dimakis  et al. established the minimum repair bandwidth 
\[
\frac{d\ell}{d - k + 1}
\]
for recovering a failed node $C_i$ from any other \(d\) helper nodes $C_j$'s, known as the cut-set bound.

\begin{definition}\label{def_MSR} 
An \((n,k,\ell)\) MDS array code \(\C\) over \(\F_q\) is said to be an \(((n,k,d),(\ell,\beta=\frac{\ell}{d-k+1}))\) \emph{Minimum Storage Regenerating (MSR)} code if for every \(i\in[n]\) , for all \(\cR\subseteq[n]\bl \{i\} \) of size \(d\), there exist  functions \(f^\cR_{i,j}:\F^\ell\rightarrow\F^\beta,j\in \cR\) and function \(g^{\cR}_{i}:\F^{d\beta}\rightarrow\F^{\ell}\), such that the code symbol \(C_i\), can be computed by \(g^\cR_{i}(f^\cR_{i,j}(C_j),j\in \cR)\).
\end{definition}

In distributed storage node repair, besides the repair bandwidth, the amount of data accessed is also crucial, as it directly affects the disk I/O overhead. The optimal-access property refers to the condition where the volume of data transmitted by a helper node is exactly the same as the amount of data accessed. MSR codes that satisfy this condition are called \emph{optimal-access} MSR codes.

Much progress has been made in deriving lower bounds on the sub-packetization $\ell$. 
For general MSR codes with repair degree \(d = n - 1\), Goparaju et al.~\cite{Goparaju2014} proved the lower bound \(\ell \geq \exp(\Theta(\sqrt{k}))\). 
This result was later improved to \(\ell \geq e^{\frac{k - 1}{4r}}\) by Alrabiah and Guruswami~\cite{Guruswami2019}, where \(r=n-k\). The most recent result, \(\ell \geq e^{\frac{(k - 1)(r - 1)}{2r^2}}\), was presented by Alrabiah and Guruswami~\cite{Guruswami2021} and Balaji et al.~\cite{balaji2022}.
For optimal-access MSR codes, Tamo et al.~\cite{tamo2014} demonstrated that \(\ell \geq r^{\frac{k - 1}{r}}\) with repair degree \(d = n - 1\). 
Balaji et al.~\cite{balaji2022} further established 
the lower bound \(\ell \geq \min\{s^{\lceil\frac{n-1}{s}\rceil}, s^{k-1}\}\) with $s=d-k+1$. 
This bound is tight as it is achieved by construction in~\cite{Sasidharan2016,Li2024,Ye2017,Kumar2023,Li2018,Myna2018}. 
In addition to the constructions mentioned above, numerous MSR code constructions have been proposed in recent years. 
Please refer to~\cite{Birenjith2015,Rashmi2011,Tamo2017,Tamo2013,Wang2016,Raviv2017,Duursma2023} and the excellent  survey~\cite{CIT-115} for further details.

\subsection{Multiple-Node Repair}\label{intro_multi}
In practical scenarios, repairing multiple nodes failures is a common situation.
For multi-node repair, two main models are commonly considered: the \emph{centralized} repair and the \emph{cooperative} repair. 
In the centralized repair model, a single data center downloads data from \(d\) helper nodes and uses this data to rebuild \(h\) failed nodes. In contrast, the cooperative repair model involves the failed nodes themselves collaborating with each other, acquiring data from helper nodes, and sharing information between them during the repair process. The repair bandwidth for these two models is calculated differently: in centralized repair, the bandwidth is determined by the total data downloaded from the helper nodes, while in cooperative repair, it includes both the data received from helper nodes and the data exchanged between the failed nodes. 
Please see~\cite{Cadambe2013,Wang2017,Zorgui2017,ye20173,Rawat2018,Zorgui2018} for results under centralized repair model and~\cite{Hu2010,Shum2013,Kermarrec2011,shum2011,Li2014,shum2016} for results under cooperative repair model. 

This paper specifically focuses on centralized MSR codes. Under this model, the cut-set bound is established by Cadambe et al. in~\cite{Cadambe2013}. 
As we generally consider linear repair, the formal definition of centralized MSR code is given as follows.
(We offer a formal proof of the cut-set bound in Appendix~\ref{proof_cutset}.)

\begin{definition}\label{def_MSR2} 
Let $\C$ be an \((n,k,\ell)\) MDS array code. Given two integers \(h,d\) with \(2 \le h \le r,\ k+1\le d \le n-h\), and a subset $\cF\subseteq[n]$ of size $h$, we say that \emph{$\C$ has  $(h,d)$ optimal repair schemes
for $\cF$} if for any \(\cR\subseteq[n]\bl \cF \) of size \(d\), there exist linear functions \(f^\cR_{\cF,j}:\F^l\rightarrow\F^\beta,j\in \cR\) and linear function \(g_{\cF}^{\cR}:\F^{d\beta}\rightarrow\F^{h\ell}\), 
such that the code symbols \(C_i\), \(i\in \cF\), can be computed by \(g_{\cF}^{\cR}(f^\cR_{\cF,j}(C_j):j\in \cR)\), where \(\beta=\frac{h\ell}{d-k+h}\).
In general, if $\C$ has $(h,d)$ optimal repair schemes for any  $\cF\subseteq[n]$ of size $h$, then we call $\C$ an \emph{$(h,d)$ centralized MSR code}.
An  \((h, d)\)-centralized MSR code with optimal-access property is called an optimal-access \((h, d)\)-centralized MSR code.
\end{definition}

As \( f^\cR_{\cF,j} \) and \(g^{\cR}_{\cF}\), \( j \in \cR \) are linear, there exist matrices \( S^{\cR}_{j \rightarrow \cF} \) of sizes \( \beta \times \ell \)  and \( M^{\cR}_{\cF}\) of size \(h\ell\times d\beta\) corresponding to \( f^\cR_{\cF,j} \)and \(g^{\cR}_{\cF}\), respectively. 
Specifically, the data transmitted from node \( C_j \), for the purpose of repairing the nodes in \( \cF \), is \( S^\cR_{j \rightarrow \cF} C_j \).
The matrices \( S^{\cR}_{j \rightarrow \cF} \), for \( j \in \cR \) are referred to as the \emph{repair matrices}, and the span of the rows of \( S^{\cR}_{j \rightarrow \cF} \), denoted by \(\langle S^{\cR}_{j \rightarrow \cF}  \rangle\), is called the \emph{repair subspace}.
If \(d = n - h\), the helper nodes necessarily comprises all remaining nodes, and we will drop the superscript \(\cR\) of  \(S^\cR_{j \rightarrow \cF}\) and \( M^{\cR}_{\cF}\).
If the repair matrices depend solely on the set of failed nodes \( \cF \), namely, \(S^{\cR}_{j \rightarrow \cF} = S_{\cF} \), for any \( \cR\subseteq[n]\bl{\cF}\) and \(j\in \cR\),  such  repair matrices are commonly referred to as  \emph{constant repair matrices}. 
If \( S^{\cR}_{j \rightarrow \cF} = S_{j \rightarrow \cF} \) for any \( \cR\subseteq[n]\bl{\cF}\), then the repair matrices are said to be \emph{independent of the choice of the remaining \((d-1)\) helper nodes}~\cite{balaji2022}. 
These two types of repair matrices play a crucial role in study of lower bounds on the sub-packetization.

Numerous constructions for both the centralized and cooperative repair models have been proposed in recent years. Please refer to~\cite{Wang2017,zorgui2019,Ye20172,li2023,zhang2024,Ye2019,Zhang2020,Ye2020,liu2023new,Zhangzh2024} for further reading.
However, to date, no lower bound has been established on the sub-packetization of MSR codes in the context of multiple-node repair. 

\subsection{Main Results}
In this paper, we establish the following novel lower bounds on the sub-packetization  of optimal-access MSR codes in the context of multiple-node failures. 
These bounds generalize the tight bounds for single-node failure presented by Balaji et al.~\cite{balaji2022}. 
\begin{theorem}\label{thm_constant1}
       Let \(\C\) be an \((n,k,\ell)\) MDS array code. If \(\C\) has  \((h,d)\) optimal-access repair schemes with constant repair matrices for any $h$ failed systematic nodes, then 
       $$\ell\ge\left(\frac{d-k+h}{h}\right)^{\left\lceil\left\lfloor\frac{k}{h}\right\rfloor\frac{h}{d-k+h}\right\rceil}.$$  
\end{theorem}

\begin{theorem}\label{thm_constant2}
     Let \(\C\) be an optimal-access \((h,d)\) centralized MSR code with constant repair matrices. Then 
     $$\ell\ge\min\left\{\left(\frac{d-k+h}{h}\right)^{\lfloor\frac{k}{h}\rfloor},\left(\frac{d-k+h}{h}\right)^{\lceil\frac{n}{d-k+h}\rceil}\right\}.$$
\end{theorem}

\begin{theorem}\label{thm_nonconstant2}
     Let \(\C\) be an optimal-access \((h,d)\) centralized  MSR code with repair matrices independent of the choice of the remaining \((d-1)\) helper nodes. Then
      \[\ell\ge\min\left\{\left(\frac{d-k+h}{h}\right)^{\lfloor\frac{k-1}{h}\rfloor},\left(\frac{d-k+h}{h}\right)^{\lceil\lfloor\frac{n-1}{h}\rfloor\frac{h}{d-k+h}\rceil}\right\}.\]
\end{theorem}
\begin{remark}
\begin{enumerate}
    \item
    In~\cref{thm_constant2} if we let $h=1$, then we get a tight bound on the sub-packetization of MSR codes with constant repair matrices. This bound is achieved by the construction in ~\cite{Sasidharan2016,Li2024,Ye2017,Kumar2023,Li2018,Myna2018}. 
    \item 
    In~\cref{thm_nonconstant2} if we let \(d = n - h\), then we can get a lower bound for optimal-access \((h, n - h)\) centralized MSR codes directly.
    \item
    As a cooperative MSR code is also a centralized MSR code by the result of Ye et al.~\cite{Ye2019}, our lower bounds also apply to cooperative MSR codes.
\end{enumerate}
\end{remark}
Furthermore, we can employ generating functions to conduct a more refined analysis.
For any \(1\le c \le \lfloor\frac{h}{2}\rfloor\), we set 
$$\Delta=\Delta(c):=\left(\frac{h-c}{d-k+h}\right)^2+\frac{4c}{d-k+h}$$
and 
$$\alpha=\alpha(c):=\frac{(c-h)+\sqrt{\Delta}(d-k+h)}{2c}.$$

\begin{theorem}\label{generating_1}
     Let \(\C\) be an optimal-access \((h,d)\) centralized MSR code with constant repair matrices. Then 
     $$\ell\ge\min\left\{ \frac{\sqrt{\Delta}}{4}\alpha^{\left(\left\lfloor\frac{k-h}{h-c}\right\rfloor+1\right)},\frac{\sqrt{\Delta}}{4}\alpha^{\left\lceil\left(\left\lfloor\frac{n-h}{h-c}\right\rfloor+1\right)\cdot\frac{h}{d-k+h}\right\rceil}\right\},$$
     for any \(1\le c \le \lfloor\frac{h}{2}\rfloor\) and $\Delta$, $\alpha$ defined above.
\end{theorem}
\begin{theorem}\label{generating_2}
     Let \(\C\) be an optimal-access \((h,d)\) centralized  {MSR} code with repair matrices independent of the choice of the remaining \((d-1)\) helper nodes. Then 
     $$l\ge\min\left\{ \frac{\sqrt{\Delta}}{4}\alpha^{\left(\left\lfloor\frac{k-1-h}{h-c}\right\rfloor+1\right)},\frac{\sqrt{\Delta}}{4}\alpha^{\left\lceil\left(\left\lfloor\frac{n-1-h}{h-c}\right\rfloor+1\right)\cdot\frac{h}{d-k+h}\right\rceil}\right\},$$
     for any \(1\le c \le \lfloor\frac{h}{2}\rfloor\) and $\Delta$, $\alpha$ defined above.
\end{theorem}
\begin{remark}
For high-rate MSR codes, and for sufficiently large \( n \), it can be shown that the bounds in Theorems~\ref{generating_1} and~\ref{generating_2} are greater than those in Theorems~\ref{thm_constant2} and~\ref{thm_nonconstant2}. Therefore, unlike the case when \( h=1 \), the generalized bounds in Theorems~\ref{thm_constant2} and~\ref{thm_nonconstant2} are generally not tight.
\end{remark}

\subsection{Outline of This Paper} 
The rest of this paper is organized as follows. In \cref{sec_linear}, we will study the linear algebraic description for centralized model, which plays a fundamental role in the subsequent proofs. In \cref{sec_constant}, we will focus on deriving the lower bounds for optimal-access \((h,d)\) centralized MSR codes and prove  \cref{thm_constant1} and \cref{thm_constant2}. In \cref{non-constant}, we will extend our analysis to the case repair matrices independent of the choice of the remaining \((d-1)\) helper nodes and prove \cref{thm_nonconstant2}. In \cref{improving_bounds}, we improve our bounds by using generating functions and prove \cref{generating_1} and \cref{generating_2}. Section~\ref{sec:conclusion} provides the conclusion of our paper.

\section{Linear Algebraic Settings for Centralized Repair}\label{sec_linear}
In this section, we consider the problem of repairing $h$ systematic nodes using the remaining $n-h$ nodes.
We will follow the notations used by Balaji~et al.~\cite{balaji2022}.
Let $\C$ be an $(n,k,\ell)$ array code over $\F_q$.
Write \(r=n-k\). Suppose that \( \cU = \{ u_1, \dots, u_k \} \) is a collection of $k$ systematic nodes of $\C$,
and let \( \cP = [n] \bl \cU=\{ p_1, \dots, p_r \} \) be the left $r$ parity nodes.
Then we can get a generator matrix \(G\) of $\C$ as
\begin{align*}
    G = \bbordermatrix{ & u_1 & \cdots &u_k \cr
              u_1&I_\ell&& \cr
              \vdots&&\ddots& \cr
              u_k&&&I_\ell\cr
              p_1&A_{p_1,u_1}&\ldots&A_{p_,u_k}\cr
              \vdots&\vdots&\ddots&\vdots\cr
              p_r&A_{p_r,u_1}&\ldots&A_{p_r,u_k}\cr
              }
\end{align*}
where \(I_\ell\) is the \(\ell\times \ell\) identity matrix and \(A_{p_i,u_j},i\in[r],j\in[k]\) are \(\ell\times \ell\) matrices over \(\F_{q}\). 
Here the rows of $G$ are indexed by the nodes $u_1, \ldots, u_k, p_1,\ldots,p_k$, and the columns by $u_1, \ldots, u_k$. 
Write
\[A=\begin{bmatrix}
    A_{p_1,u_1}&\ldots&A_{p_1,u_k}\\
    \vdots&\ddots&\vdots\\
    A_{p_r,u_1}&\ldots&A_{p_r,u_k}
\end{bmatrix}.\]
It is well-known that $\C$ is an MDS array code if and only if every square block submatrix of \(A\) is invertible.

The following notations will be used throughout this paper.
For any subset \(\cE=\{u_{i_1},\ldots,u_{i_t}\}\subseteq \cU\), let \( A(:,\cE) \) be the submatrix of $A$ consisting of
the block columns indexed by $\{u_{i_1},\ldots,u_{i_t}\}$. 
For any $\cF\subseteq [n]$ of size $h\ (2\leq h\leq r)$, $\cR\subseteq [n]\backslash \cF$ of size $d\ (k+1\leq d\leq n-h)$, 
and its subset  \(\mathcal{H}=\{{i_1},\ldots,{i_p}\}\subseteq\cR\), we write 
\[S^\cR_{\mathcal{H}\rightarrow\cF}=\di( S^\cR_{i_1\rightarrow \cF},\ldots,S^\cR_{i_p\rightarrow \cF}).\]
If \(d=n-h\), we can drop the superscript \(\cR\). For a matrix \(B\), we denote its row space by \(\langle B \rangle\). 

Now we give a linear algebraic description of the optimal repair schemes for $h$ failed systematic nodes using the remaining $n-h$ nodes.
\begin{lemma}\label{algebraic_setting}
Let $\C$ be an $(n,k,\ell)$ MDS array code over $\F_q$. Suppose that for any $h$ failed systematic nodes \(\cF=\{u_{i_1},\ldots,u_{i_h}\}\subseteq \cU\), there exists a \((h,d)\) optimal repair scheme for $\cF$ with \(\cR=[n]\backslash \cF\), that is,  
there exist repair matrices \( S_{j \rightarrow \cF} \) of sizes \( \beta \times \ell \), $j\in[n]\bl \cF$, and a matrix \(M_\cF\) of size \( h\ell \times d\beta \), such that 
\begin{equation}\label{optimalrepir}
    \begin{bmatrix}
    C_{u_{i_1}}\\
    \vdots\\
    C_{u_{i_h}}
\end{bmatrix}=M_\cF\begin{bmatrix}
    S_{u_{i_{h+1}}\rightarrow \cF}C_{u_{i_{h+1}}}\\
    \vdots\\
    S_{u_{i_k}\rightarrow \cF}C_{u_{i_k}}\\
    S_{p_{1}\rightarrow \cF}C_{p_{1}}\\
    \vdots\\
    S_{p_{r}\rightarrow \cF}C_{p_{r}}
\end{bmatrix} \ \text{for any} \; C=(C_1,\ldots,C_n)^{T}\in\mathcal{C},
\end{equation}
where we write \(\cU \bl \cF=\{u_{i_{h+1}},\ldots,u_{i_{k}}\}\).
Then we have 
     \begin{enumerate}
         \item 
         The square matrix
         \(S_{\cP\rightarrow\cF}A(:, \cF)\)
        is of full rank.\\
         \item For $1\leq t\leq r$ and $h+1\leq z\leq k$, we have
         \(\langle S_{p_t\rightarrow \cF}A_{p_t,u_{i_z}}\rangle=\langle S_{u_{i_z}\rightarrow \cF}\rangle\).\\
          \item For any $j\in[n]\bl \cF$, the repair matrix \( S_{j\rightarrow\cF} \) has full row rank.
     \end{enumerate}
    \end{lemma}
   \begin{proof}
Write
\[M_\cF=\begin{bmatrix}
    M_{1,u_{i_{h+1}}}&\ldots&M_{1,u_{i_k}}&M_{1,p_{1}}&\ldots&M_{1,p_{r}}\\
    \vdots&&\vdots&\vdots&&\vdots\\
M_{h,u_{i_{h+1}}}&\ldots&M_{h,u_{i_k}}&M_{h,p_{1}}&\ldots&M_{h,u_{p_r}} 
\end{bmatrix}\]
as an $h\times d$ block matrix, where each block entry \(M_{i,j}\) is of size \(\ell\times \beta\). Let \(M_{\cF}(:, \cU\bl\cF)\) be the submatrix of $M_{\cF}$ consisting of the first \(k-h\) block columns and \(M_{\cF}(:, \cP)\) the submatrix consisting of the last  \(r\) block columns. Then
$M_\cF=[ M_{\cF}(:, \cU\bl\cF)\ M_{\cF}(:, \cP)].$ Similarly write
\begin{align*}    
A=\begin{bmatrix}
A_{p_1,u_{i_{1}}}&\ldots&A_{p_1,u_{i_h}}&A_{p_1,u_{i_{h+1}}}&\ldots&A_{p_1,u_{i_k}}\\
\vdots&\ddots&\vdots&\vdots &&\vdots&\\
A_{p_r,u_{i_{1}}}&\ldots&A_{p_r,u_{i_h}}&A_{p_r,u_{i_{h+1}}}&\ldots&A_{p_r,u_{i_k}}
\end{bmatrix}
=[A(:, \cF)\ A(:, \cU\bl\cF)].
\end{align*}
Recall that
\begin{align*}
\begin{bmatrix}
    C_{p_{1}}\\
    \vdots\\
    C_{p_{r}}
\end{bmatrix}
&=\begin{bmatrix}
A_{p_1,u_{i_{1}}}&\ldots&A_{p_1,u_{i_h}}&A_{p_1,u_{i_{h+1}}}&\ldots&A_{p_1,u_{i_k}}\\
\vdots&\ddots&\vdots&\vdots &&\vdots&\\
A_{p_r,u_{i_{1}}}&\ldots&A_{p_r,u_{i_h}}&A_{p_r,u_{i_{h+1}}}&\ldots&A_{p_r,u_{i_k}}
\end{bmatrix}
\begin{bmatrix}
    C_{u_{i_1}}\\
    \vdots\\
    C_{u_{i_h}}\\
    C_{u_{i_{h+1}}}\\
    \vdots\\
    C_{u_{i_k}}
\end{bmatrix}\\
&=\begin{bmatrix}
A(:, \cF)\ A(:, \cU\bl\cF)
\end{bmatrix}
\begin{bmatrix}
    C_{u_{i_1}}\\
    \vdots\\
    C_{u_{i_h}}\\
    C_{u_{i_{h+1}}}\\
    \vdots\\
    C_{u_{i_k}}
\end{bmatrix}. 
\end{align*}
By (\ref{optimalrepir}) we have
\begin{align*}
    \begin{bmatrix}
    C_{u_{i_1}}\\
    \vdots\\
    C_{u_{i_h}}
\end{bmatrix}
&=
M_\cF
\begin{bmatrix}
  S_{\cU\bl\cF\rightarrow\cF}&\bO\\ 
  \\
  \bO&S_{P\rightarrow\cF}
\end{bmatrix}
\begin{bmatrix}
  \bO&I_{(k-h)\ell}\\
  \\
A(:, \cF)&A(:, \cU \setminus \cF)
\end{bmatrix}
\begin{bmatrix}
    C_{u_{i_1}}\\
    \vdots\\
    C_{u_{i_h}}\\
    C_{u_{i_{h+1}}}\\
    \vdots\\
    C_{u_{i_k}}
\end{bmatrix}  \\
& =
[M_{\cF}(:, \cU\bl\cF)\ M_{\cF}(:, P)]
\begin{bmatrix}
    \mathbf{O}&{S}_{ \cU\bl\cF\rightarrow\cF}\\
    \\
S_{P\rightarrow\cF}A(:, \cF)&S_{P\rightarrow\cF}A(:, \cU \setminus \cF)\\
\end{bmatrix}
\begin{bmatrix}
    C_{u_{i_1}}\\
    \vdots\\
    C_{u_{i_h}}\\
    C_{u_{i_{h+1}}}\\
    \vdots\\
    C_{u_{i_k}}
\end{bmatrix}\\
&= M_{\cF}(:, P)S_{P\rightarrow\cF}A(:, \cF)\begin{bmatrix}
    C_{u_{i_1}}\\
    \vdots\\
    C_{u_{i_h}}\\
    \end{bmatrix}+\left(M_{\cF}(:, \cU\bl\cF){S}_{ \cU\bl\cF\rightarrow\cF}+M_{\cF}(:, P)S_{P\rightarrow\cF}A(:, \cU \setminus \cF)\right)\begin{bmatrix}
    C_{u_{i_{h+1}}}\\
    \vdots\\
    C_{u_{i_k}}
\end{bmatrix}.
\end{align*}
Set \(C_{u_{i_{h+1}}}=\ldots=C_{u_{i_k}}=0\). Then we obtain
   \begin{align*}
       \begin{bmatrix}
    C_{u_{i_1}}\\
    \vdots\\
    C_{u_{i_h}}\\
\end{bmatrix}=M_{\cF}(:, P)S_{P\rightarrow\cF}A(:, \cF)
       \begin{bmatrix}
    C_{u_{i_1}}\\
    \vdots\\
    C_{u_{i_h}}\\
\end{bmatrix}, 
   \end{align*}
   for \([C_{u_{i_1}},\ldots,C_{u_{i_h}}]^T\in \F_q^{h\ell}\). 
   Hence \(M_{\cF}(:, \cP)S_{\cP\rightarrow\cF}A(:, \cF)=I_{h\ell}\). 
   It follows directly that the matrices
   \(S_{p_t\rightarrow \cF}\), \(1 \le t\le r\), have full row rank, and
   the square matrices $M_{\cF}(:, \cP)$ and $S_{\cP\rightarrow\cF}A(:, \cF)$ both have full rank.

 Set \(C_{u_{i_{1}}}=\ldots=C_{u_{i_h}}=0\). Then we obtain 
   \begin{align*}
       \left(M_{\cF}(:, \cU\bl\cF)S_{\cU\bl\cF\rightarrow\cF}+M_{\cF}(:, \cP)S_{\cP\rightarrow\cF}A(:,\cU\bl\cF)\right)\begin{bmatrix}
    C_{u_{i_{h+1}}}\\
    \vdots\\
    C_{u_{i_k}}\\
\end{bmatrix}=0, 
   \end{align*}
   for all  \([C_{u_{i_{h+1}}},\ldots,C_{u_{i_k}}]^T\in \F_q^{(k-h)l}\). Hence
   \begin{align*}
      S_{\cP\rightarrow\cF}A(:,\cU\bl\cF)&=-(M_{\cF}(:, \cP))^{-1}M_{\cF}(:, \cU\bl\cF)S_{\cU\bl\cF\rightarrow\cF}.
   \end{align*}
   Comparing the block columns of the above two matrices, we have
   \begin{align*}
       \begin{bmatrix}
       S_{p_1\rightarrow \cF}A_{p_1,u_{i_{z}}}\\
       \vdots\\
       S_{p_r\rightarrow \cF}A_{p_r,u_{i_{z}}}
   \end{bmatrix}=-(M_{\cF}(:, \cP))^{-1}\begin{bmatrix}
       M_{1,u_{i_z}}\\
       \vdots\\
       M_{h,u_{i_z}}
   \end{bmatrix}S_{u_{i_{z}}\rightarrow \cF}
   \end{align*}
    for $h+1\leq z\leq k$.
  It is not hard to check that \(\langle S_{p_t\rightarrow \cF}A_{p_t,u_{i_z}}\rangle\subseteq\langle S_{u_{i_z}\rightarrow \cF}\rangle\) for \(1\leq t\leq r\), $h+1\leq z\leq k$. 
 As the matrices \(S_{p_t\rightarrow \cF}\) and \(A_{p_t,u_{i_z}}\) both have full rank, we get \(\dim\langle S_{p_t\rightarrow \cF}A_{p_t,u_{i_z}}\rangle=\dim\langle S_{u_{i_z}\rightarrow \cF}\rangle=\beta\).
 Therefore \(\langle S_{p_t\rightarrow \cF}A_{p_t,u_{i_z}}\rangle=\langle S_{u_{i_z}\rightarrow \cF}\rangle\), and 
 the repair matrix \(S_{u_{i_z}\rightarrow \cF}\) has full row rank.
\end{proof}

The following key lemma will be used in latter sections to derive our new lower bounds.
\begin{lemma} \label{key_lem}
Let \(\cF_1=\{u_{i_1},\ldots,u_{i_h}\}\) and \(\cF_2=\{v_{i_1},\ldots,v_{i_h}\}\) be two subsets of \( \cU \).
Then we have
\[\rank \left(S_{P\rightarrow\cF_2}A(:, \cF_1)\right)\le\frac{(h-|\cF_1 \cap \cF_2|)h\ell}{r}+|\cF_1 \cap \cF_2|\ell.\]
Moreover, if \(\cF_1\cap \cF_2=\emptyset\), then the equality holds, that is,
\begin{align*}
\rank \left(S_{P\rightarrow\cF_2}A(:, \cF_1)\right)=\frac{h^2\ell}{r},
\end{align*}
and
         \begin{equation}\label{key_lem_eq1}
            \langle  S_{P\rightarrow\cF_2}A(:, \cF_1)\rangle=\langle
                S_{\cF_1\rightarrow \cF_2}\rangle.
        \end{equation}
\end{lemma}
\begin{proof}
Write $t=|\cF_1 \cap \cF_2|.$
Without loss of generality, we can assume that $\cF_1 \cap \cF_2=\{u_{i_{h-t+1}},\ldots,u_{i_h}\}$ and \(u_{i_z}=v_{i_z}\) for $h-t+1\leq z\leq h$. Write
  \[S_{P\rightarrow\cF_2}A(:, \cF_1)=\begin{bmatrix}
      S_{P\rightarrow\cF_2}A(:, \cF_1\bl (\cF_1\cap\cF_2))&
  S_{P\rightarrow\cF_2}A(:,\cF_1\cap\cF_2)
  \end{bmatrix}
  ,\]
  and note that
  \begin{align*}
    \left\langle S_{P\rightarrow\cF_2}A(:, \cF_1)\right\rangle
    \subseteq
    \left\langle\begin{bmatrix}
    S_{\cF_1\bl (\cF_1\cap\cF_2)\rightarrow\cF_2}&\bO\\[8pt]
    S_{P\rightarrow\cF_2}A(:, \cF_1\bl (\cF_1\cap\cF_2))&S_{P\rightarrow\cF_2}A(:,\cF_1\cap\cF_2)\\
\end{bmatrix}\right\rangle. 
 \end{align*}
  By 2) of  \cref{algebraic_setting} we have 
  \(\langle S_{p_{t}\rightarrow\cF_2}A_{p_t,u_{i_z}}\rangle=\langle S_{u_{i_z}\rightarrow\cF_2}\rangle\), for \(1\le t\le r\) and \(1\le z\le h-t\). 
  Using elementary row operations, it is not hard to check that
  \begin{align*}
    \left\langle\begin{bmatrix}
    S_{\cF_1\bl (\cF_1\cap\cF_2)\rightarrow\cF_2}&\bO\\[8pt]
    S_{P\rightarrow\cF_2}A(:, \cF_1\bl (\cF_1\cap\cF_2))&S_{P\rightarrow\cF_2}A(:,\cF_1\cap\cF_2)\\
\end{bmatrix}\right\rangle
=
    \left\langle\begin{bmatrix}
    S_{\cF_1\bl (\cF_1\cap\cF_2)\rightarrow\cF_2}&\bO\\[8pt]
    \bO&S_{P\rightarrow\cF_2}A(:,\cF_1\cap\cF_2)\\
\end{bmatrix}\right\rangle. 
  \end{align*}
Therefore
 \begin{align}\label{key_lem_2}
    \langle S_{P\rightarrow\cF_2}A(:, \cF_1)\rangle\subseteq
    \left\langle\begin{bmatrix}
    S_{\cF_1\bl (\cF_1\cap\cF_2)\rightarrow\cF_2}&\bO\\[8pt]
    \bO&S_{P\rightarrow\cF_2}A(:,\cF_1\cap\cF_2)\\
\end{bmatrix}\right\rangle. 
 \end{align}
By 3) of  \cref{algebraic_setting} we have $\rank(S_{\cF_1\bl (\cF_1\cap\cF_2)\rightarrow\cF_2})=\frac{(h-t)h\ell}{r}$.
On the other hand, $\rank(S_{P\rightarrow\cF_2}A(:,\cF_1\cap\cF_2))$ is at most $t\ell$, the number of its columns.
Hence 
 \begin{align*}
     \rank (S_{P\rightarrow\cF_2}A(:, \cF_1))&\le \rank(S_{\cF_1\bl (\cF_1\cap\cF_2)\rightarrow\cF_2})+\rank(S_{P\rightarrow\cF_2}A(:,\cF_1\cap\cF_2))\\
     &\le \frac{(h-t)h\ell}{r}+t\ell. 
 \end{align*}

Now if \(\cF_1\cap\cF_2=\emptyset\), then~\eqref{key_lem_2} becomes
\[
\langle S_{P\rightarrow\cF_2}A(:, \cF_1)\rangle\subseteq \langle S_{\cF_1\rightarrow \cF_2}\rangle. 
\]
 By  3) of  \cref{algebraic_setting} we see that the matrix \(S_{P\rightarrow\cF_2}\) has full row rank,
 and $\rank(S_{\cF_1\rightarrow \cF_2})=\frac{h^2\ell}{r}$.
 By the MDS property of $\C$ we know that the first \(h\) block rows of matrix \(A(:, \cF_1)\) is invertible. 
 Then 
 \[\rank \left(S_{P\rightarrow\cF_2}A(:, \cF_1)\right)\ge\rank\left(S_{\{p_1,\ldots,p_h\}\rightarrow\cF_2}\right)=\frac{h^2\ell}{r}.\]
It follows directly that 
            $\langle  S_{P\rightarrow\cF_2}A(:, \cF_1)\rangle=\langle
                S_{\cF_1\rightarrow \cF_2}\rangle.$

\end{proof}

\section{Lower Bounds with Constant Repair Matrices}\label{sec_constant}
In this section,  we derive lower bounds on the sub-packetization of optimal-access \((h,d)\)-centralized MSR codes with constant repair matrices, namely,  \(S^{\cR}_{j\rightarrow \cF}=S_{\cF}\) for \(\cF\subseteq [n]\), \(\cR\subseteq [n]\bl\cF\) and \(j\in\cR\). 
The proof method is first employed by Tamo et al.~\cite{tamo2014}, and later developed by Balaj et al.~\cite{balaji2022}.
Before we proceed to the proofs we need the following result.
\begin{lemma}\label{lem_constant1}
  Suppose that $\C$ has  $(h,d)$ optimal repair schemes with constant repair matrices for any $h$ failed systematic nodes.
  Let $1\leq t\leq \lfloor\frac{k}{h}\rfloor$, and let \(\cF_1, \dots, \cF_t \subseteq \cU\) be any $t$ pairwise disjoint subsets of size $h$ of $k$ systematic nodes. 
  Then we have 
    \[
    \dim\left(\bigcap\limits_{i=1}^{t} \langle S_{\cF_i} \rangle\right) \leq \left(\frac{h}{d-k+h}\right)^t\ell.
    \]
\end{lemma}
\begin{proof}
We first assume that \(d = n - h\) and prove the result by induction on \(t\). 
For $t=1$, by 3) of \cref{algebraic_setting}, we have \( \dim(\langle S_{\cF_1} \rangle) = \frac{h\ell}{r}\). 
Assume the result holds for \(t\), and now we will prove it also holds for \(t+1\).
Fix a basis for the subspace \(\bigcap\limits_{i=1}^{t+1}\langle S_{\cF_{i}}\rangle\),
and set $S$ to be the matrix whose rows are basis vectors.
Then \( \langle S \rangle =\bigcap\limits_{i=1}^{t+1}\langle S_{\cF_{i}}\rangle\), 
and in particular \(\langle S\rangle\subseteq \langle S_{\cF_i}\rangle\) for \(i\in [t+1]\). 
Hence we can find matrices \(T_i\), whose rank is equal to $\dim(\langle S\rangle)$,  such that \(S=T_iS_{\cF_i}\).
It follows that
\begin{align*}
   (I_{r}\otimes S)A{(:,\cF_1)}&=(I_{r}\otimes T_i)(I_{r}\otimes S_{\cF_i})A{(:,\cF_1)}, \;\text{for}\;i\in[t+1].
\end{align*}
 If $i=1$, then by 1) of \cref{algebraic_setting} we see that  the matrix \((I_{r}\otimes S_{\cF_1})A{(:,\cF_1)}\) is invertible. 
Then
\begin{align}\label{lem_constant1_1}
  \dim(\langle (I_{r}\otimes S)A{(:,\cF_1)}\rangle)=\dim(\langle I_{r}\otimes T_1\rangle)=r\dim(\langle S\rangle).  
\end{align}
On the other hand, for \(2\leq i\leq t+1\)  we have 
$$\langle (I_{r}\otimes S)A{(:,\cF_1)}\rangle\subseteq \langle (I_{r}\otimes S_{\cF_i})A{(:,\cF_1)}\rangle=\langle I_h\otimes S_{\cF_i} \rangle.$$ 
Here the latter equaltiy follows from (\ref{key_lem_eq1}) of \cref{key_lem}. 
 Then
\begin{align*}
    \langle (I_{r}\otimes S)A{(:,\cF_1)}\rangle \subseteq\mathop{\bigcap}_{i=2}^{t+1}\langle (I_{r}\otimes S_{\cF_i})A{(:,\cF_1)}\rangle=\mathop{\bigcap}_{i=2}^{t+1}\langle I_h\otimes S_{\cF_i} \rangle.
\end{align*}
By assumption we have
\begin{align}
\begin{split}
\label{lem_constant1_2}
\dim\left(\langle (I_{r}\otimes S)A{(:,\cF_1)}\rangle\right)
&\le \dim\left(\mathop{\bigcap}_{i=2}^{t+1}\langle I_h\otimes S_{\cF_i} \rangle\right)\\
&= h\cdot\dim \left(\mathop{\bigcap}\limits_{i=2}^{t+1}\langle S_{\cF_i}\rangle\right)\\
&\le h\left(\frac{h}{r}\right)^{t}\ell.
\end{split}
\end{align}
By (\ref{lem_constant1_1}) and (\ref{lem_constant1_2}), we can compute that
\begin{align*}
    \dim(\langle S\rangle) \leq \left(\frac{h}{r}\right)^{t+1}\ell.
\end{align*}
This concludes our proof for the case \(d=n-h\).

For general \(d\), we first puncture any $n-d-h$ parity nodes from $\C$ and get a new code $\C'$ with $n'=d+h$ and $r'=d-k+h$.
Applying the above argument to $\C'$ we get 
$$\dim\left(\bigcap\limits_{i=1}^{t} \langle S_{\cF_i} \rangle\right) \leq \left(\frac{h}{d-k+h}\right)^t\ell.$$
\end{proof}

\begin{proof}[Proof of Theorem \ref{thm_constant1}]\label{proof_constant1}
  Let \(\{e_j : j \in [\ell]\}\) be the standard basis of \(\F_q^\ell\), 
  where $e_j$ denotes the vector with a $1$ in the $j$th coordinate and $0$'s elsewhere.
  Use \(\cU=\{u_{1},\ldots,u_{k}\} \) to denote the $k$ systematic nodes, and 
  set \( \cF_i=\{u_{(i-1)h+1},\ldots,u_{ih}\}\) for $1\leq i\leq \lfloor\frac{k}{h}\rfloor$.
  Define a bipartite graph where one set of vertices is the standard basis vectors 
\(\{e_j : j \in [\ell]\}\), and the other set of vertices is the constant repair 
subspaces \(\langle S_{\mathcal{F}_i} \rangle\) for \(1 \leq i \leq \lfloor \frac{k}{h} \rfloor\). 
An edge exists between a vertex \(e_j\) and a subspace \(\langle S_{\mathcal{F}_i} \rangle\) if and only if 
\(e_j \in \langle S_{\mathcal{F}_i} \rangle\). 

Now, count the total number of edges in this bipartite graph using two different methods.
From the optimal-access property of the repair schemes we see that 
  there are exactly \(\frac{h\ell}{d-k+h}\) nonzero columns in the repair matrix $S_{\cF_i}$.
  Hence the degree of each repairing subspace in the graph is \(\frac{h\ell}{d-k+h}\). 
  Thus, the total number of edges in the graph is \(\lfloor\frac{k}{h}\rfloor\cdot\frac{h\ell}{d-k+h}\).  
  On the other hand, let \(\mathcal{N}_j=\{i:e_j\in\langle S_{\cF_i}\rangle\}\) be the index set of repairing subspaces which contain \(e_j\). 
  Then by \cref{lem_constant1},
     \begin{align}\label{numberN_1}
         1\leq\dim\left(\bigcap_{i\in \mathcal{N}_j }\langle S_{\cF_i}\rangle\right)\leq \left(\frac{h}{d-k+h}\right)^{|\mathcal{N}_j|}\ell.
     \end{align}
  It follows that the degree of each standard basis vector in the graph satisfies 
  \begin{align}\label{numberN_2}
     |\mathcal{N}_j|\leq \left\lfloor\log_{\frac{d-k+h}{h}}\ell\right\rfloor.
  \end{align}
  Hence there are at most \(\ell\cdot\left\lfloor\log_{\frac{d-k+h}{h}}\ell\right\rfloor\) edges in the graph.
  In summary we have
  \begin{align*}
     \left\lfloor\frac{k}{h}\right\rfloor\cdot\frac{h\ell}{d-k+h}\le\ell\cdot \left\lfloor\log_{\frac{d-k+h}{h}}\ell\right\rfloor,
  \end{align*}
  and 
  \begin{align*}
  \ell\ge\left(\frac{d-k+h}{h}\right)^{\left\lceil\left\lfloor\frac{k}{h}\right\rfloor\frac{h}{d-k+h}\right\rceil}.
  \end{align*}
  \end{proof}


\begin{proof}[Proof of Theorem \ref{thm_constant2}]
  If $\ell\geq \left(\frac{d-k+h}{h}\right)^{\left\lfloor\frac{k}{h}\right\rfloor},$ then the conclusion holds.
  Hence we assume that $\ell<\left(\frac{d-k+h}{h}\right)^{\left\lfloor\frac{k}{h}\right\rfloor}$.
  Set \( \cF_i=\{{(i-1)h+1},\ldots,{ih}\}\) for $1\leq i\leq \left\lfloor\frac{n}{h}\right\rfloor$.
  Define a bipartite graph where one set of vertices is the standard basis vectors 
\(\{e_j : j \in [\ell]\}\), and the other set of vertices is the constant repair 
subspaces $\langle S_{\cF_i}\rangle ,1\leq i\leq \left\lfloor\frac{n}{h}\right\rfloor$. 
An edge exists between a vertex \(e_j\) and a subspace \(\langle S_{\cF_i}\rangle\) if and only if \(e_j\in \langle S_{\cF_i}\rangle\). 

Count the total number of edges in this bipartite graph using two different methods.
  Similar to the proof of Theorem \ref{thm_constant1}, there are exactly\(\left\lfloor\frac{n}{h}\right\rfloor\cdot\frac{h\ell}{d-k+h}\) edges in the graph by optimal-access property. 
  On the other hand, we first prove that for all \(\cE\subset \{1,\ldots,\lfloor\frac{n}{h}\rfloor\}\),
  \begin{align}\label{proof_thm_constant2_1}
     \dim\left(\mathop{\bigcap}\limits_{i\in \cE}\langle S_{ \cF_i}\rangle\right)\leq \left(\frac{h}{d-k+h}\right)^{|\cE|}. 
  \end{align}
    For $|\cE|\leq \lfloor\frac{k}{h}\rfloor$, we can always choose a collection of $k$ systematic nodes which contains $\bigcup\limits_{i\in \cE} \cF_i$.  Then the result follows directly from Lemma~\ref{lem_constant1}.
  For \(\left\lfloor\frac{k}{h}\right\rfloor<|\cE|\le \left\lfloor\frac{n}{h}\right\rfloor\), we choose any subset $\mathcal{H}$ of $\cE$ with size $\lfloor\frac{k}{h}\rfloor$, and get 
     \begin{align*}
       \dim\left(\bigcap_{i\in \cE} \langle S_{\cF_i} \rangle\right)
       &\leq \dim\left(\bigcap_{i\in\mathcal{H}} \langle S_{\cF_i} \rangle\right) \\
       &\leq \left(\frac{h}{d-k+h}\right)^{\left\lfloor\frac{k}{h}\right\rfloor}\ell\\
       &<1.
    \end{align*}
  Therefore (\ref{proof_thm_constant2_1}) holds for \(1\leq |\cE|\leq\left\lfloor\frac{n}{h}\right\rfloor\). 
    Just as (\ref{numberN_1}) and (\ref{numberN_2}), there are at most \(\ell\cdot\left\lfloor\log_{\frac{d-k+h}{h}}\ell\right\rfloor\) edges in the graph.
   We have
    \begin{align*}
     \lfloor\frac{n}{h}\rfloor\cdot\frac{h\ell}{d-k+h}\le\ell\cdot \left\lfloor\log_{\frac{d-k+h}{h}}\ell\right\rfloor,
    \end{align*}
   and 
    \begin{align*}
    \ell\ge\left(\frac{d-k+h}{h}\right)^{\left\lceil\left\lfloor\frac{n}{h}\right\rfloor\frac{h}{d-k+h}\right\rceil}.
    \end{align*}
   In summary we have
     \[l\ge\min\left\{\left(\frac{d-k+h}{h}\right)^{\left\lfloor\frac{k}{h}\right\rfloor},\left(\frac{d-k+h}{h}\right)^{\left\lceil\left\lfloor\frac{n}{h}\right\rfloor\frac{h}{d-k+h}\right\rceil}\right\}.\]
\end{proof}


\section{A Lower Bound with Repair Matrices Independent of Identity of Remaining  Helper Nodes}\label{non-constant}
In this section we consider optimal-access \((h,d)\) centralized MSR codes with repair matrices 
independent of the choice of the remaining \((d-1)\) helper nodes, namely, \(S_{j\rightarrow\cF}^{\cR}=S_{j\rightarrow\cF}\), for \(\cF\subseteq [n]\), \(\cR\subseteq [n]\bl\cF\) and \(j\in\cR\).

\begin{lemma}\label{lem_nonconstant1}
Suppose that $\C$ is an optimal-access \((h,n-h)\)-centralized MSR code. 
Let $1\leq t\leq \left\lfloor\frac{k-1}{h}\right\rfloor$, and let \(\cF_1, \dots, \cF_t \subseteq [n]\) be $t$ pairwise disjoint subsets of $[n]$ of size \(h\). 
Then for any \(j \in [n] \setminus \bigcup\limits_{i=1}^{t} \cF_i \), we have
  \begin{align}\label{lem_nonconstant1_1}
        \dim\left(\bigcap_{i=1}^{t} \langle S_{j\rightarrow\cF_i} \rangle\right) \leq \left(\frac{h}{r}\right)^t\ell.
  \end{align}
\end{lemma}
\begin{proof}
First we choose a collection of $k-th$ nodes $\cV$ from $[n] \setminus  \bigcup\limits_{i=1}^{t} \cF_i $.
Now we regard the nodes in $\cU= \bigcup\limits_{i=1}^{t} \cF_i \bigcup \cV$ as the systematic nodes,
and the other nodes $\cP=\{p_1,\ldots,p_r\}=[n]\backslash \cU$ the parity nodes. Then we have \([n] \setminus  \bigcup\limits_{i=1}^{t} \cF_i =\cV \cup \cP\).

For any \(v\in \cV\) and \(p\in \cP\), we have
   \begin{align}
       \dim\left(\mathop{\bigcap}\limits_{i=1}^{t} \langle S_{p \rightarrow \cF_i} \rangle\right)
       &=\dim\left( \mathop{\bigcap}\limits_{i=1}^{t} \langle S_{p \rightarrow \cF_i} A_{p,v}\rangle \right)\label{subsec_nonconstant1_2}\\
       &=\dim\left(\mathop{\bigcap}\limits_{i=1}^{t}\langle S_{v\rightarrow \cF_i}\rangle\right)\label{subsec_nonconstant1_3}.
   \end{align}
Here (\ref{subsec_nonconstant1_2}) follows from that \(A_{p,v}\) is invertible and (\ref{subsec_nonconstant1_3}) follows from 2) of \cref{algebraic_setting}. 
Therefore the dimension of the subspaces $\mathop{\bigcap}\limits_{i=1}^{t} \langle S_{{j} \rightarrow \cF_i} \rangle$
are the same for all nodes $j\in\cV \cup \cP$.
 
Then we  prove the result  by induction on \(t\). 
For \(t = 1\), by 3) of \cref{algebraic_setting}, we have 
$$ \dim(\langle S_{j\rightarrow\cF_1} \rangle) = \frac{h}{r}\ell \quad\text{for } j\in\cV \cup \cP.$$
Assume the result holds for \(t\), and now we will prove it also holds for \(t+1\). From above, it suffices to  prove the result for \(j=p_1\).

For any node \(p\in \cP\), fix a basis for the subspace \(\bigcap\limits_{i=1}^{t+1}\langle S_{p\rightarrow\cF_{i}}\rangle\),
and set $S_p$ to be the matrix whose rows are the basis vectors.
Then \( \langle S_p \rangle =\bigcap\limits_{i=1}^{t+1}\langle S_{p\rightarrow\cF_{i}}\rangle\), and 
\(\langle S_p\rangle\subseteq \langle S_{p\rightarrow\cF_i}\rangle\) for \(i\in [t+1]\). 
Hence we can find matrices \(T_{p,i}\), whose rank is equal to $\dim(\langle S_p \rangle)$, such that \(S_p=T_{p,i}S_{p\rightarrow\cF_i}\).
It follows that
\begin{align*}
   \begin{bmatrix}
        S_{p_1}&&\\
        &\ddots&\\
        &&S_{p_r}
    \end{bmatrix}A{(:,\cF_1)}&=\begin{bmatrix}
        T_{p_1,i}&&\\
        &\ddots&\\
        &&T_{p_r,i}
    \end{bmatrix}S_{P\rightarrow \cF_i}A{(:,\cF_1)},
\end{align*}
for $i\in[t+1]$. If $i=1$, then by 1) of \cref{algebraic_setting} we see that  the matrix \(S_{P\rightarrow \cF_1}A{(:,\cF_1)}\) is invertible. 
Then
\begin{align}\label{lem_nonconstant1_2}
\begin{split}
  \dim\left\langle 
  \begin{bmatrix}
        S_{p_1}&&\\
        &\ddots&\\
        &&S_{p_r}
    \end{bmatrix}A{(:,\cF_1)}
    \right\rangle
    &=\dim\left\langle
    \begin{bmatrix}
        T_{p_1,1}&&\\
        &\ddots&\\
        &&T_{p_r,1}
    \end{bmatrix}
    \right\rangle\\
    &=\sum_{i=1}^{r}\rank(T_{p_i,1})\\
    &=r\cdot \dim(\langle S_{p_1} \rangle).\\
    \end{split}
\end{align}
On the other hand, for \(2\leq i\leq t+1\)  we have 
$$\left\langle \begin{bmatrix}
        S_{p_1}&&\\
        &\ddots&\\
        &&S_{p_r}
    \end{bmatrix}A{(:,\cF_1)}\right\rangle\subseteq \left\langle S_{P\rightarrow \cF_i}A{(:,\cF_1)}\right\rangle=\left\langle S_{\cF_1\rightarrow\cF_i} \right\rangle.$$ 
Here the latter equality follows from (\ref{key_lem_eq1}) of \cref{key_lem}. 
 Then
\begin{align*}
    \left\langle \begin{bmatrix}
        S_{p_1}&&\\
        &\ddots&\\
        &&S_{p_r}
    \end{bmatrix}A{(:,\cF_1)}\right\rangle \subseteq\mathop{\bigcap}_{i=2}^{t+1}\langle  S_{P\rightarrow\cF_i}A{(:,\cF_1)}\rangle=\mathop{\bigcap}_{i=2}^{t+1}\langle  S_{\cF_1\rightarrow\cF_i} \rangle.
\end{align*}
By assumption we have
\begin{align}
\begin{split}
\label{lem_nonconstant1_3}
\dim\left(\left\langle \begin{bmatrix}
        S_{p_1}&&\\
        &\ddots&\\
        &&S_{p_r}
    \end{bmatrix}A{(:,\cF_1)}\right\rangle\right)
&\le \dim\left(\mathop{\bigcap}_{i=2}^{t+1}\langle  S_{\cF_1\rightarrow\cF_i} \rangle\right)\\
&= h\cdot\dim \left(\mathop{\bigcap}\limits_{i=2}^{t+1}\langle S_{u\rightarrow\cF_i}\rangle\right)\\
&\le h\left(\frac{h}{r}\right)^{t}\ell,
\end{split}
\end{align}
for any \(u\in \cF_1\).
By (\ref{lem_nonconstant1_2}) and (\ref{lem_nonconstant1_3}), we can compute that
\begin{align*}
     \dim(\langle S_{p_1}\rangle)\leq \left(\frac{h}{r}\right)^{t+1}\ell.
\end{align*}

 \end{proof}

\begin{proof}[Proof of  Theorem \ref{thm_nonconstant2}]
If \(\ell\ge\left(\frac{d-k+h}{h}\right)^{\left\lfloor\frac{k-1}{h}\right\rfloor}\), then the conclusion holds.
Hence we assume that \(\ell<\left(\frac{d-k+h}{h}\right)^{\left\lfloor\frac{k-1}{h}\right\rfloor}\).
  Set \( \cF_i=\{{(i-1)h+1},\ldots,{ih}\}\) for $1\leq i\leq \left\lfloor\frac{n-1}{h}\right\rfloor$.
  Define a bipartite graph where one set of vertices is the standard basis vectors 
\(\{e_j : j \in [\ell]\}\), and the other set of vertices is repair 
subspaces $\langle S_{n\rightarrow\cF_i}\rangle ,1\leq i\leq \left\lfloor\frac{n-1}{h}\right\rfloor$. 
An edge exists between a vertex \(e_j\) and a subspace \(\langle S_{n\rightarrow\cF_i}\rangle\)  if and only if 
\(e_j\in \langle S_{n\rightarrow\cF_i}\rangle\). 

Count the total number of edges using two different methods.
  Similar to the proof of \cref{thm_constant1}, there are exactly
 \(\left\lfloor\frac{n-1}{h}\right\rfloor\cdot\frac{h\ell}{d-k+h}\) edges in the graph. 
  On the other hand,
  we first prove that for all \(\cE\subset \{1,\ldots,\lfloor\frac{n-1}{h}\rfloor\}\) and \(|\cE|\le \left\lfloor\frac{k-1}{h}\right\rfloor\), 
  \[\dim\left(\mathop{\bigcap}\limits_{i\in \cE}\langle S_{n\rightarrow \cF_i}\rangle\right)\leq \left(\frac{h}{d-k+h}\right)^{|\cE|}.\] 
  It holds for \(d=n-h\) by Lemma \ref{lem_nonconstant1}. 
  For general \(d\), we puncture some $n-d-h$ nodes from $[n-1]\bl\bigcup\limits_{i\in \cE}\cF_i$, and get a new code $\C'$ with length $n'=d+h$, dimension \(k'=k\) and \(d=n'-h\).  
  The result follows by applying Lemma~\ref{lem_nonconstant1} to $\C'$.
  
  For \(\left\lfloor\frac{k-1}{h}\right\rfloor<|\cE|\le \left\lfloor\frac{n-1}{h}\right\rfloor\), we choose any subset $\mathcal{H}$ of $\cE$ with size $\lfloor\frac{k-1}{h}\rfloor$, and obtain
    \begin{align*}
        \dim\left(\bigcap_{i\in\cE}\langle S_{n\rightarrow\cF_i}\rangle\right)&\leq \dim\left(\bigcap_{i\in\mathcal{H}} \langle S_{n\rightarrow\cF_i} \rangle\right) \\
        &\leq \left(\frac{h}{d-k+h}\right)^{\left\lfloor\frac{k-1}{h}\right\rfloor}\ell\\&<1.
    \end{align*}
    Therefore,
    \begin{align*}
        \dim\left(\bigcap_{i\in \cE} \langle S_{n\rightarrow\cF_i} \rangle\right)\leq\left(\frac{h}{d-k+h}\right)^{|\cE|}\ell \quad \text{for }1\le|\cE|\le \left\lfloor\frac{n-1}{h}\right\rfloor.
    \end{align*}
    Let \(\mathcal{N}_j=\{i:e_j\in\langle S_{n\rightarrow\cF_i}\rangle\}\) be the index set of repairing subspaces which contain \(e_j\). 
   Just as (\ref{numberN_1}) and (\ref{numberN_2}), there are at most \(\ell\cdot\left\lfloor\log_{\frac{d-k+h}{h}}\ell\right\rfloor\) edges in the graph.
   We have
    \begin{align*}
     \left\lfloor\frac{n-1}{h}\right\rfloor\cdot\frac{h\ell}{d-k+h}\le\ell\cdot \left\lfloor\log_{\frac{d-k+h}{h}}\ell\right\rfloor,
    \end{align*}
   and 
    \begin{align*}
    \ell\ge\left(\frac{d-k+h}{h}\right)^{\left\lceil\left\lfloor\frac{n-1}{h}\right\rfloor\frac{h}{d-k+h}\right\rceil}.
    \end{align*}
   In summary we have
     \[l\ge\min\left\{\left(\frac{d-k+h}{h}\right)^{\left\lfloor\frac{k-1}{h}\right\rfloor},\left(\frac{d-k+h}{h}\right)^{\left\lceil\left\lfloor\frac{n-1}{h}\right\rfloor\frac{h}{d-k+h}\right\rceil}\right\}.\]
\end{proof}

 \begin{corollary}
     Let \(\C\) be an optimal-access \((h,n-h)\) centralized MSR code. Then
      \[l\ge\min\left\{\left(\frac{r}{h}\right)^{\left\lceil\frac{k-1}{h}\right\rceil},\left(\frac{r}{h}\right)^{\left\lceil\left\lfloor\frac{n-1}{h}\right\rfloor\frac{h}{r}\right\rceil}\right\}.\]
 \end{corollary}\label{thm_nonconstant1}
 \begin{proof}
     The result follows directly by setting $d=n-h$ in~\cref{thm_nonconstant2}.
 \end{proof}
    
 \section{Improving Bounds Using Generating Functions}\label{improving_bounds}
 In this section, we employ generating functions to perform a more refined analysis, further enhancing the bounds in Sections \ref{sec_constant}-\ref{non-constant}.
 \subsection{Generating Functions}
    Fix an integer \(\co\) with \(1\le c<h\).  Let \(a\) (\(2\le a\le h\)) be the unique integer  such that  \(c \in \left(\frac{(a-2)h}{a-1},\frac{(a-1)h}{a}\right]\). 
 We define a strictly decreasing sequence \(\{B_n\}_{n\ge 0}\) by the following recurrence formula
    \begin{align}\label{generating_function_key}
        \begin{cases} 
     B_0=1, \\[5pt]
      B_1=\frac{h}{d-k+h},\\[5pt]
     B_j= \frac{h-c}{d-k+h}B_{j-1}+\ldots+\frac{h-c}{d-k+h}B_{1}+\frac{(j-1)c-(j-2)h}{d-k+h}B_0,\ j\in [a-1],\\[5pt]
     B_{n+a}=\frac{h-c}{d-k+h}B_{n+a-1}+\ldots+\frac{h-c}{d-k+h}B_{n+1}+\frac{(a-1)c-(a-2)h}{d-k+h}B_n, n\in \mathbb{N}.
     \end{cases}
    \end{align}
 Let \(B=\sum\limits_{n\ge 0} B_nx^n\) be the generating function of the above recurrence formula. Then it is not hard to check that
 \begin{align*}
     \frac{B-\sum\limits_{i=0}^{a-1}B_ix^i}{x^a}=\frac{h-c}{d-k+h}\frac{B-\sum\limits_{i=0}^{a-2}B_ix^i}{x^{a-1}}+\ldots+\frac{h-c}{d-k+h}\frac{B-B_0}{x}+\frac{(a-1)c-(a-2)h}{d-k+h}B.
 \end{align*}
 Via direct calculation, we get
 \begin{align}
     B=\frac{\sum\limits_{i=0}^{a-1}B_ix^i-\frac{h-c}{d-k+h}\sum\limits_{i=0}^{a-2}B_ix^{i+1}-\ldots-\frac{h-c}{d-k+h}B_0x^{a-1}}{1-\frac{h-c}{d-k+h}x-\ldots-\frac{h-c}{d-k+h}x^{a-1}-\frac{(a-1)c-(a-2)h}{d-k+h}x^a}.
 \end{align}
 Write 
 \begin{align}
     \phi(x)= \theta_1 x^a+\theta_2 x^{a-1}+\cdots+\theta_2 x+1, 
 \end{align}
 where \(\theta_1=-\frac{(a-1)c-(a-2)h}{d-k+h}\), \(\theta_2=\frac{c-h}{d-k+h}\).
 By calculating  roots of \(\phi(x)\) over \(\mathbb{C}\), we can expand \(B\) as a formal power series and then get the expression of \(B_n\). 
 In particular if the polynomial \(\phi(x)\) has only simple roots \(x_1,\ldots,x_a\), namely, 
 $$\phi(x)=\theta_1\prod\limits_{i=1}^{a} (x-x_i),$$
  we can compute that
  \begin{align}\label{generalB}
   B=\frac{1}{\theta_1}\sum_{n\ge0}\left(\sum\limits_{i=1}^a\left(\frac{1}{x_{i}^{n-a+2}\prod\limits_{j\neq i }(x_i-x_j)}\left(\sum\limits_{u=0}^{a-1}\left(\frac{-1}{x_i^{a-u-1}}+\frac{h-c}{d-k+h}\sum\limits_{v=0}^{a-u-2}\frac{1}{x_i^{v}}\right)B_u\right)\right)\right)x^n.
  \end{align}
 If \(1\le c \le \lfloor\frac{h}{2}\rfloor\), then we have \(a=2\), and
 \begin{align*}
     B=\frac{B_0+(B_1-\frac{h-c}{d-k+h}B_0)x}{1-\frac{h-c}{d-k+h}x-\frac{c}{d-k+h}x^2}.
 \end{align*}
 Now the polynomial \(\phi(x)\) is  quadratic and \(\Delta=\Delta(c):=\left(\frac{h-c}{d-k+h}\right)^2+\frac{4c}{d-k+h}> 0\), 
 hence $\phi(x)$ has two different roots:
 \[\alpha=\alpha(c):=\frac{(c-h)+\sqrt{\Delta}(d-k+h)}{2c},\ \gamma=\gamma(c):=\frac{(c-h)-\sqrt{\Delta}(d-k+h)}{2c}. \]
 By (\ref{generalB}), we have 
 \begin{align*}
     B&=\frac{1}{\sqrt\Delta}\sum\limits_{n\ge0}\left(\left(\frac{1}{\alpha}+\frac{c}{d-k+h}\right)\frac{1}{\alpha^n}+\left(-\frac{1}{\gamma}-\frac{c}{d-k+h}\right)\frac{1}{\gamma^n}\right) x^n.
 \end{align*}
 As \(1<|\alpha|<|\gamma|\), we get
 \begin{align}\label{a_2b_n}
 \begin{split}
     B_n &\le\frac{2}{\sqrt\Delta}\left(\frac{1}{\alpha}+\frac{c}{d-k+h}\right)\frac{1}{\alpha^n}\\
     &\le \frac{4}{\sqrt\Delta}\frac{1}{\alpha^n}.
 \end{split}
  \end{align}
 
 \subsection{A Lower Bound with Constant  Repair Matrices}
First, we consider repair schemes with constant repair matrices.  
The following result generalizes~\cref{lem_constant1} by allowing the subsets \(\mathcal{F}_i\) to have nonempty intersections.
 \begin{lemma}\label{generating_lemma_1}
     Suppose that $\C$ has  $(h,d)$ optimal repair schemes with constant repair matrices. 
     Fix an integer \(\co\) with \(1\le c<h\).  Let \(a\) (\(2\le a\le h\)) be the unique  integer  such that  \(c \in \left(\frac{(a-2)h}{a-1},\frac{(a-1)h}{a}\right]\). 
     Let \(\cF_i=\{(i-1)(h-c)+1,\ldots,(i-1)(h-c)+h\}\), \(1\le i\le\lfloor\frac{n-h}{h-c}\rfloor+1\). 
     Then, for any \(\cE\subseteq [\lfloor\frac{n-h}{h-c}\rfloor+1]\) with \( 1\le |\cE|\le \lfloor\frac{k-h}{h-c}\rfloor+1\), we have 
     \begin{equation}
        \dim\left(\bigcap\limits_{i\in \cE}\langle S_{\cF_i}\rangle\right)\leq B_{|\cE|}\ell,
    \end{equation}
    where \(B_{|\cE|}\) is the \(|\cE|\)-th  term of the sequence \(\{B_n\}_{n\ge 0}\) defined above.
 \end{lemma}
 \begin{proof}
 Write \(t=|\cE|\). We first assume that $\cE$ is consisting of some $t$ consecutive integers.
 Without loss of generality, we assume that systematic nodes set \(\cU=[k]\) and \(\cE = \{1, 2, \dots, t\}\).

  If \(d = n - h\), we can prove the result by induction on \(t\). 
  For \( t = 1 \), the conclusion holds as \( \dim(\langle S_{\cF_1} \rangle) = \frac{h\ell}{r} \).
  Assume that the the result holds for \(t\), and now we will prove it also holds for \(t+1\).
Just as in the proof of~\cref{lem_constant1}, there exist  matrices \(S\) and \(T_i\) of the same rank such that \( \langle S \rangle =\bigcap\limits_{i=1}^{t+1}\langle S_{\cF_{i}}\rangle\) and \(S=T_iS_{\cF_i}\) for \(i\in[t+1]\).
It follows that
\begin{align*}
   (I_{r}\otimes S)A{(:,\cF_1)}&=(I_{r}\otimes T_i)(I_{r}\otimes S_{\cF_i})A{(:,\cF_1)},\;\text{for}\;i\in[t+1].
\end{align*}
If $i=1$, then by 1) of \cref{algebraic_setting} we see that  the matrix \((I_{r}\otimes S_{\cF_1})A{(:,\cF_1)}\) is invertible. 
Then
\begin{align}\label{generating_function_1}
  \dim(\langle (I_{r}\otimes S)A{(:,\cF_1)}\rangle)=\dim(\langle I_{r}\otimes T_1\rangle)=r\dim(\langle S\rangle).  
\end{align}
On the other hand, for \(2\leq i\leq t+1\),
 we have \(\langle (I_{r}\otimes S)A{(:,\cF_1)}\rangle\subseteq \langle (I_{r}\otimes S_{\cF_i})A{(:,\cF_1)}\rangle\). Then by  
  (\ref{key_lem_2}) of \cref{key_lem}, we get
 \begin{align*}
     &\left\langle (I_{r}\otimes S_{\cF_i})A{(:,\cF_1)}\right\rangle\subseteq\left\langle\begin{bmatrix}
     I_{|\cF_1\bl (\cF_1\cap\cF_i)|}\otimes S_{\cF_i}&\bO\\[8pt]
     \bO&(I_{r}\otimes S_{\cF_i})A(:,\cF_1\cap\cF_i)\\
 \end{bmatrix}\right\rangle. \\
 &\subseteq\left\langle\begin{bmatrix}
      I_{|\cF_1\bl (\cF_1\cap\cF_i)|}\otimes S_{\cF_i}&\bO
 \end{bmatrix}\right\rangle+
 \left\langle\begin{bmatrix}
     \bO&(I_{r}\otimes S_{\cF_i})A(:,\cF_1\cap\cF_i)
 \end{bmatrix}\right\rangle.
  \end{align*}
 So, we can obtain
 \begin{align}\label{generating_function_2}
 \begin{split}
     &\left\langle (I_{r}\otimes S)A{(:,\cF_1)}\right\rangle \subseteq\mathop{\bigcap}_{i=2}^{t+1}\left\langle\begin{bmatrix}
     I_{|\cF_1\bl (\cF_1\cap\cF_i)|}\otimes S_{\cF_i}&\bO\\[8pt]
     \bO&(I_{r}\otimes S_{\cF_i})A(:,\cF_1\cap\cF_i)\\
 \end{bmatrix}\right\rangle\\
 &\subseteq\mathop{\bigcap}_{i=2}^{t+1}\left(\left\langle\begin{bmatrix}
      I_{|\cF_1\bl (\cF_1\cap\cF_i)|}\otimes S_{\cF_i}&\bO
 \end{bmatrix}\right\rangle+
 \left\langle\begin{bmatrix}
     \bO&(I_{r}\otimes S_{\cF_i})A(:,\cF_1\cap\cF_i)
 \end{bmatrix}\right\rangle\right).
 \end{split}
 \end{align}
     
     If \(t\in [a-1]\),  then we see that the intersection \(\cF_j\cap\cF_1\) is nonempty for \(j\in \{2,\ldots,t+1\}\), as 
     \[(j-1)(h-c)+1<(a-1)\left(h-\frac{(a-2)h}{a-1}\right)+1=h+1.\]
     By (\ref{generating_function_2}) we have 
     \begin{align*}
     &\dim\langle (I_{r}\otimes S)A{(:,\cF_1)}\rangle \le (h-c)\sum_{i=2}^{t+1}\dim\left(\mathop{\bigcap}_{j=i}^{t+1} \langle
      S_{\cF_j}\rangle\right)+(tc-(t-1)h)\ell.\\
 \end{align*}
 Combine (\ref{generating_function_1}) and (\ref{generating_function_2}), we can compute that
 \begin{align*}
     \dim(\langle S\rangle)&\le \frac{h-c}{r}\sum_{i=1}^{t}B_{i}\ell
      +\frac{tc-(t-1)h}{r}B_0\ell\\
      &=\left(\frac{h-c}{r}\sum_{i=1}^{t}B_{i}
      +\frac{tc-(t-1)h}{r}B_0\right)\ell\\
      &=B_{t+1}\ell,
 \end{align*}
  and the conclusion holds.
 
 If \(t\ge a\),  then we see that the intersection \(\cF_j\cap\cF_1\) is empty for \(j\geq a+1\), as 
    \[(j-1)(h-c)+1\ge a\left(h-\frac{(a-1)h}{a}\right)+1=h+1.\]
     We have 
    \begin{align}\label{generating_function_3}
    \dim\langle (I_{r}\otimes S)A{(:,\cF_1)}\rangle \le (h-c)\sum_{i=2}^{a}\dim\left(\mathop{\bigcap}_{j=i}^{t+1} \langle
      S_{\cF_i}\rangle\right)+\left(((a-1)c-(a-2)h)\mathop{\bigcap}_{i=a+1}^{t+1} \langle
      S_{\cF_i}\rangle\right).
 \end{align}
 Combine (\ref{generating_function_1}) and (\ref{generating_function_3}), we can compute that
 \begin{align*}
     \dim(\langle S\rangle)&\le \frac{h-c}{r}\sum_{i=t-a+2}^{t}B_{i}\ell
      +\frac{(a-1)c-(a-2)h}{r}B_{t+1-a}\ell\\
      &= \left(\frac{h-c}{r}\sum_{i=t-a+2}^{t}B_{i}
      +\frac{(a-1)c-(a-2)h}{r}B_{t+1-a}\right)\ell\\&=B_{t+1}.
 \end{align*}
 This concludes our proof for the case \(d=n-h\). For general \(d\), we first puncture some $n-d-h$ parity nodes from $\C$ and get a new code $\C'$ with $n'=d+h$ and $r'=d-k+h$.
 Applying the above argument to $\C'$ we get the conclusion.

 For the case that $\cE$ does not consist of 
\(t\) consecutive integers, it is not hard to prove the result similarly as above, and we omit the details.

 \end{proof}
 
  \begin{proof}[Proof of \cref{generating_1}]
  If $\ell\geq \frac{\sqrt{\Delta}}{4}\alpha^{(\lfloor\frac{k-h}{h-c}\rfloor+1)},$ then the conclusion holds.
  Hence we assume that  $\ell< \frac{\sqrt{\Delta}}{4}\alpha^{(\lfloor\frac{k-h}{h-c}\rfloor+1)}$.
  Set \(\cF_i=\{(i-1)(h-c)+1,\ldots,(i-1)(h-c)+h\}\) for \(1\le j\le\lfloor\frac{n-h}{h-c}\rfloor+1\).
  Define a bipartite graph where one set of vertices is the standard basis vectors 
\(\{e_j : j \in [\ell]\}\), and the other set of vertices is the constant repair 
subspaces $\langle S_{\cF_i}\rangle ,1\le i\le\lfloor\frac{n-h}{h-c}\rfloor+1$. 
An edge exists between a vertex \(e_j\) and a subspace \(\langle S_{\cF_i}\rangle\) if and only if \(e_j\in \langle S_{\cF_i}\rangle\). 

Count the total number of edges in this bipartite graph using two different methods.
  Similar to the proof of \cref{thm_constant1}, there are exactly
  \(\left(\lfloor\frac{n-h}{h-c}\rfloor+1\right)\cdot\frac{h\ell}{d-k+h}\) edges in the graph.
  On the other hand, we first prove that for any \(\cE\subseteq [\lfloor\frac{n-h}{h-c}\rfloor+1]\),
  \begin{equation}\label{proof_generating_function_1}
        \dim(\bigcap\limits_{i\in \cE}\langle S_{\cF_i}\rangle)\leq \frac{4}{\sqrt\Delta}\frac{1}{\alpha^n}\ell.
    \end{equation}
    By \((\ref{a_2b_n})\), we have 
 \begin{align*}
     B_n\le\frac{4}{\sqrt\Delta}\frac{1}{\alpha^n}.
  \end{align*}
  For \(\left(\lfloor\frac{k-h}{h-c}\rfloor+1\right)<|\cE|\le \left(\lfloor\frac{n-h}{h-c}\rfloor+1\right)\), we choose any subset $\mathcal{H}$ of $\cE$ with size $\lfloor\frac{k-h}{h-c}\rfloor+1$. By~\cref{generating_lemma_1}, we obtain
    \begin{align*}
       \dim\left(\bigcap_{i\in\cE} \langle S_{\cF_i} \rangle\right)
       &\leq \dim\left(\bigcap_{i\in\mathcal{H} }^{} \langle S_{\cF_i} \rangle\right) \\
       &\le\frac{4}{\sqrt\Delta}\frac{1}{\alpha^{(\lfloor\frac{k-h}{h-c}\rfloor+1)}}\ell\\
       &<1.\\
    \end{align*}
  Combining this with~\cref{generating_lemma_1}, we conclude that (\ref{proof_generating_function_1}) holds for
  \(\cE\subseteq [\lfloor\frac{n-h}{h-c}\rfloor+1]\).
    Let \(\mathcal{N}_j=\{i:e_j\in\langle S_{\cF_i}\rangle\}\) be the index set of repairing subspaces which contain \(e_j\). 
   Then we have
     \begin{align}\label{general_generating_1}
         1\leq\dim\left(\bigcap_{i\in \mathcal{N}_j }\langle S_{\cF_i}\rangle\right)\le\frac{4}{\sqrt\Delta}\frac{1}{\alpha^{|\mathcal{N}_j|}}\ell.
     \end{align}
    As \(\alpha\ge 1\), it follows that 
    \begin{align}\label{general_generating_2}
     |\mathcal{N}_j|\leq \left\lfloor\log_{\alpha}\left(\frac{4}{\sqrt\Delta}\ell\right)\right\rfloor.
    \end{align}
   Hence there are at most \(\ell\cdot\left\lfloor\log_{\alpha}\left(\frac{4}{\sqrt\Delta}\ell\right)\right\rfloor\) edges in the graph.
   We have
    \begin{align*}
     \left(\left\lfloor\frac{n-h}{h-c}\right\rfloor+1\right)\cdot\frac{h\ell}{d-k+h}\le\ell\cdot\left\lfloor\log_{\alpha}\left(\frac{4}{\sqrt\Delta}\ell\right)\right\rfloor,
    \end{align*}
   and 
    \begin{align*}
    \ell\ge\frac{\sqrt{\Delta}}{4}\alpha^{\left\lceil\left(\lfloor\frac{n-h}{h-c}\rfloor+1\right)\cdot\frac{h}{d-k+h}\right\rceil}.
    \end{align*}
   In summary we have
     \[l\ge\min\left\{ \frac{\sqrt{\Delta}}{4}\alpha^{\left(\lfloor\frac{k-h}{h-c}\rfloor+1\right)},\frac{\sqrt{\Delta}}{4}\alpha^{\left\lceil\left(\lfloor\frac{n-h}{h-c}\rfloor+1\right)\cdot\frac{h}{d-k+h}\right\rceil}\right\}.\]
\end{proof}
 
\subsection{A Lower Bound with Repair Matrices Independent of Identity of Remaining  Helper Nodes}
In this section we consider optimal-access \((h,d)\) centralized MSR codes with repair matrices 
independent of the choice of the remaining \((d-1)\) helper nodes. This approach generalizes the bound in ~\cref{lem_constant1}.
 \begin{lemma}\label{generating_lemma_2}
     Let $\C$ be an optimal-access \((h,d)\)-centralized MSR code with repair matrices 
independent of the choice of the remaining \((d-1)\) helper nodes.  Fix an integer \(\co\) with \(1\le c<h\).  Let \(a\) (\(2\le a\le h\)) be the unique integer  such that  \(c \in \left(\frac{(a-2)h}{a-1},\frac{(a-1)h}{a}\right]\). Let \(\cF_i=\{(i-1)(h-c)+1,\ldots,(i-1)(h-c)+h\}\), \(1\le i\le\lfloor\frac{n-1-h}{h-c}\rfloor+1\). Then, for any \(\cE\subseteq [\lfloor\frac{n-1-h}{h-c}\rfloor+1]\) and \( 1\le |\cE|\le \lfloor\frac{k-1-h}{h-c}\rfloor+1\), we have 
     \begin{equation}
        \dim(\bigcap\limits_{i\in \cE}\langle S_{n\rightarrow\cF_i}\rangle)\leq B_{|\cE|}\ell,
    \end{equation}
    where \(B_{|\cE|}\) is the \(|\cE|\)-th  term of the sequence \(\{B_n\}_{n\ge 0}\) defined in (\ref{generating_function_key}).
 \end{lemma}
 \begin{proof}
 Write \(t=|\cE|\).
 First we choose a collection of $k-(t-1)(h-c)-h$ nodes $\cV$ from $[n-1] \setminus  \bigcup\limits_{i\in\cE} \cF_i $.
Now we regard the nodes in $\cU= \bigcup\limits_{i\in\cE} \cF_i \bigcup \cV$ as the systematic nodes,
and the other nodes $\cP=\{p_1,\ldots,p_r\}=[n]\backslash \cU$ the parity nodes. Note that \(n\in\cP\) and \([n] \setminus  \bigcup\limits_{i\in\cE}  \cF_i =\cV \cup \cP\).

For any \(v\in \cV\) and \(p\in \cP\), we have
   \begin{align}
       \dim\left(\mathop{\bigcap}\limits_{i\in\cE}  \langle S_{p \rightarrow \cF_i} \rangle\right)
       &=\dim\left( \mathop{\bigcap}\limits_{i\in\cE}  \langle S_{p \rightarrow \cF_i} A_{p,v}\rangle \right)\label{generating_function2_key_1}\\
       &=\dim\left(\mathop{\bigcap}\limits_{i\in\cE} \langle S_{v\rightarrow \cF_i}\rangle\right)\label{generating_function2_key_2}.
   \end{align}
Here (\ref{generating_function2_key_1}) follows from that \(A_{p,v}\) is invertible and (\ref{generating_function2_key_2}) follows from 2) of \cref{algebraic_setting}. 
Therefore the dimension of the subspaces $\mathop{\bigcap}\limits_{i\in\cE}  \langle S_{{j} \rightarrow \cF_i} \rangle$
are the same for all nodes $j\in\cV \cup \cP$.

Then we  prove the result  by induction on \(t\). 
 First, we  prove the lemma holds for $\cE$ when its  elements are consecutive. Without loss of generality, assume \(\cE = \{1, 2, \dots, t\}\). 

  Assume that \(d = n - h\) and prove the result by induction on \(t\). 
  For \( t = 1 \), the conclusion holds as \( \dim(\langle S_{n\rightarrow\cF_1} \rangle) = \frac{h\ell}{r} \).
  Suppose the result holds for \(t\), and  we will prove it also holds for \(t+1\).
  Just as in the proof of~\cref{lem_nonconstant1}, there exist  matrices \(S_p\) and \(T_{p,i}\) of the same rank such that \( \langle S_p \rangle =\bigcap\limits_{i=1}^{t+1}\langle S_{p\rightarrow\cF_{i}}\rangle\) and \(S_p=T_{p,i}S_{p\rightarrow\cF_i}\) for \(p\in\cP,i\in[t+1]\).
It follows that
\begin{align*}
   \begin{bmatrix}
        S_{p_1}&&\\
        &\ddots&\\
        &&S_{p_r}
    \end{bmatrix}A{(:,\cF_1)}&=\begin{bmatrix}
        T_{p_1,i}&&\\
        &\ddots&\\
        &&T_{p_r,i}
    \end{bmatrix}S_{P\rightarrow \cF_i}A{(:,\cF_1)},
\end{align*}
for $i\in[t+1]$. If $i=1$, then by 1) of \cref{algebraic_setting} we see that  the matrix \(S_{P\rightarrow \cF_1}A{(:,\cF_1)}\) is invertible. 
Then
\begin{align}\label{generating_function2_key_3}
\begin{split}
  \dim\left\langle 
  \begin{bmatrix}
        S_{p_1}&&\\
        &\ddots&\\
        &&S_{p_r}
    \end{bmatrix}A{(:,\cF_1)}
    \right\rangle
    &=\dim\left\langle
    \begin{bmatrix}
        T_{p_1,1}&&\\
        &\ddots&\\
        &&T_{p_r,1}
    \end{bmatrix}
    \right\rangle\\
    &=\sum_{i=1}^{r}\rank(T_{p_i,1})\\
    &=r\cdot \dim(\langle S_{n} \rangle)\\
    \end{split}
\end{align}
On the other hand, for \(2\leq i\leq t+1\) we have 
 \begin{align*}
     \left\langle \begin{bmatrix}
        S_{p_1}&&\\
        &\ddots&\\
        &&S_{p_r}
    \end{bmatrix}A{(:,\cF_1)}\right\rangle&\subseteq \left\langle S_{P\rightarrow \cF_i}A{(:,\cF_1)}\right\rangle,\\
    &\subseteq \left\langle\begin{bmatrix}
    S_{\cF_1\bl (\cF_1\cap\cF_i)\rightarrow\cF_i}&\bO\\[8pt]
    \bO&S_{P\rightarrow\cF_i}A(:,\cF_1\cap\cF_i)\\
\end{bmatrix}\right\rangle.
 \end{align*}
 Here the latter equality follows from  (\ref{key_lem_2}) of \cref{key_lem}. So, we can obtain
 \begin{align}\label{generating_function2_key_4}
 \begin{split}
     &\left\langle \begin{bmatrix}
        S_{p_1}&&\\
        &\ddots&\\
        &&S_{p_r}
\end{bmatrix}A{(:,\cF_1)}\right\rangle\subseteq\mathop{\bigcap}_{i=2}^{t+1}\left\langle\begin{bmatrix}
    S_{\cF_1\bl (\cF_1\cap\cF_i)\rightarrow\cF_i}&\bO\\[8pt]
    \bO&S_{P\rightarrow\cF_i}A(:,\cF_1\cap\cF_i)\\
\end{bmatrix}\right\rangle\\
 &\subseteq\mathop{\bigcap}_{i=2}^{t+1}\left(\left\langle\begin{bmatrix}
      S_{\cF_1\bl (\cF_1\cap\cF_i)\rightarrow\cF_i}&\bO
 \end{bmatrix}\right\rangle+
 \left\langle\begin{bmatrix}
     \bO&S_{P\rightarrow\cF_i}A(:,\cF_1\cap\cF_i)
 \end{bmatrix}\right\rangle\right).
 \end{split}
 \end{align}
 As shown in the proof of~\cref{generating_lemma_1}, we consider the following two cases.
 \begin{enumerate}
     \item If \(t\in [a-1]\), we have 
     \begin{align*}
     &\dim\left\langle \begin{bmatrix}
        S_{p_1}&&\\
        &\ddots&\\
        &&S_{p_r}
\end{bmatrix}A{(:,\cF_1)}\right\rangle \le (h-c)\sum_{i=2}^{t+1}\dim\left(\mathop{\bigcap}_{j=i}^{t+1} \langle
      S_{n\rightarrow\cF_j}\rangle\right)+(tc-(t-1)h)\ell.\\
 \end{align*}
 Combine (\ref{generating_function2_key_3}) and (\ref{generating_function2_key_4}), we can compute that
 \begin{align*}
     \dim(\langle S_n\rangle)&\le \frac{h-c}{r}\sum_{i=1}^{t}B_{i}
      \ell+\frac{tc-(t-1)h}{r}B_0\ell\\
      &=\left(\frac{h-c}{r}\sum_{i=1}^{t}B_{i}
      +\frac{tc-(t-1)h}{r}B_0\right)\ell\\
      &=B_{t+1}\ell,
 \end{align*}
 and the conclusion holds.
 \item If \(t\ge a\), we have 
    \begin{equation}\label{generating_function2_key_5}
    \dim\left\langle \begin{bmatrix}
        S_{p_1}&&\\
        &\ddots&\\
        &&S_{p_r}
\end{bmatrix}A{(:,\cF_1)}\right\rangle \le (h-c)\sum_{i=2}^{a}\dim\left(\mathop{\bigcap}_{j=i}^{t+1} \langle
      S_{n\rightarrow\cF_i}\rangle\right)+\left((a-1)c-(a-2)h)\mathop{\bigcap}_{i=a+1}^{t+1} \langle
      S_{n\rightarrow\cF_i}\rangle\right).
 \end{equation}
 Combine (\ref{generating_function2_key_3}) and (\ref{generating_function2_key_5}), we can compute that
 \begin{align*}
     \dim(\langle S_n\rangle)&\le \frac{h-c}{r}\sum_{i=t-a+2}^{t}B_{i}\ell
      +\frac{(a-1)c-(a-2)h}{r}B_{t+1-a}\ell,\\
      &=\left(\frac{h-c}{r}\sum_{i=t-a+2}^{t}B_{i}
      +\frac{(a-1)c-(a-2)h}{r}B_{t+1-a}\right)\ell\\
      &=B_{t+1}\ell.
 \end{align*}
 \end{enumerate}
 Thus, we have completed the proof for the case \(d = n - h\).
 For general \(d\), we puncture some $n-d-h$ nodes from $[n-1]\bl\bigcup\limits_{i\in \cE}\cF_i$, and get a new code $\C'$ with length $n'=d+h$, dimension \(k'=k\) and \(d=n'-h\).  
Applying the above argument to $\C'$ we get the conclusion.

For the case where \(\cE\) does not consist of \(t\) consecutive integers, the result can be proven in a similar manner, and we omit the details.

 \end{proof}

  \begin{proof}[Proof of \cref{generating_2}]
  If $\ell\geq \frac{\sqrt{\Delta}}{4}\alpha^{(\lfloor\frac{k-1-h}{h-c}\rfloor+1)},$ the conclusion holds.
  Hence we assume that  $\ell< \frac{\sqrt{\Delta}}{4}\alpha^{(\lfloor\frac{k-1-h}{h-c}\rfloor+1)}$.
  Let \(\{e_j : j \in [\ell]\}\) be the standard basis of \(\F_q^\ell\). 
  Set \(\cF_i=\{(i-1)(h-c)+1,\ldots,(i-1)(h-c)+h\}\)for \(1\le j\le\lfloor\frac{n-1-h}{h-c}\rfloor+1\).
  Define a bipartite graph where one set of vertices is the standard basis vectors 
\(\{e_j : j \in [\ell]\}\), and the other set of vertices is the  repair 
subspaces $\langle S_{n\rightarrow\cF_i}\rangle ,1\le i\le\left(\lfloor\frac{n-1-h}{h-c}\rfloor+1\right)$.  
An edge exists between a vertex \(e_j\) and a subspace \(\langle S_{n\rightarrow\cF_i}\rangle\) if and only if \(e_j\in \langle S_{n\rightarrow\cF_i}\rangle\). 

Count the total number of edges in this bipartite graph using two different methods.
 Similar to the proof of \cref{thm_constant1}, there are \(\left(\lfloor\frac{n-1-h}{h-c}\rfloor+1\right)\cdot\frac{h\ell}{d-k+h}\) edges in the graph.
  We next prove that for any \(\cE\subseteq [\lfloor\frac{n-1-h}{h-c}\rfloor+1]\),
  \begin{equation}\label{proof_generating_function_2}
        \dim(\bigcap\limits_{i\in \cE}\langle S_{n\rightarrow\cF_i}\rangle)\leq \frac{4}{\sqrt\Delta}\frac{1}{\alpha^n}\ell.
    \end{equation}
  By \((\ref{a_2b_n})\), we have 
 \begin{align}
     B_n\le\frac{4}{\sqrt\Delta}\frac{1}{\alpha^n}.
  \end{align}
  For \(\left(\lfloor\frac{k-h-1}{h-c}\rfloor+1\right)<|\cE|\le \left(\lfloor\frac{n-1-h}{h-c}\rfloor+1\right)\), we choose any subset $\mathcal{H}$ of $\cE$ with size $\lfloor\frac{k-h-1}{h-c}\rfloor+1$. By~\cref{generating_lemma_2},we obtain
    \begin{align*}
       \dim\left(\bigcap_{i\in\cE} \langle S_{n\rightarrow\cF_i} \rangle\right)
       &\leq \dim\left(\bigcap_{ i\in\mathcal{H}}\langle S_{n\rightarrow\cF_i} \rangle\right) \\
       &\le\frac{4}{\sqrt\Delta}\frac{1}{\alpha^{(\lfloor\frac{k-1-h}{h-c}\rfloor+1)}}\ell\\
       &<1.\\
    \end{align*}
  Combining this with~\cref{generating_lemma_2}, we conclude that (\ref{proof_generating_function_2}) holds for
  \(\cE\subseteq [\lfloor\frac{n-1-h}{h-c}\rfloor+1]\).
  Just as (\ref{general_generating_1}) and (\ref{general_generating_2}), there are at most  \(\ell\left\lfloor\log_{\alpha}\left(\frac{4}{\sqrt\Delta}\ell\right)\right\rfloor\) edges in the graph.
   We have
    \begin{align*}
     \left(\left\lfloor\frac{n-1-h}{h-c}\right\rfloor+1\right)\cdot\frac{h\ell}{d-k+h}\le\ell\cdot\left\lfloor\log_{\alpha}\left(\frac{4}{\sqrt\Delta}\ell\right)\right\rfloor,
    \end{align*}
   and 
    \begin{align*}
    \ell\ge\frac{\sqrt{\Delta}}{4}\alpha^{\left\lceil\left(\lfloor\frac{n-1-h}{h-c}\rfloor+1\right)\cdot\frac{h}{d-k+h}\right\rceil}.
    \end{align*}
   In summary we have
     \[l\ge\min\left\{ \frac{\sqrt{\Delta}}{4}\alpha^{(\lfloor\frac{k-1-h}{h-c}\rfloor+1)},\frac{\sqrt{\Delta}}{4}\alpha^{\left\lceil\left(\lfloor\frac{n-1-h}{h-c}\rfloor+1\right)\cdot\frac{h}{d-k+h}\right\rceil}\right\}.\]
\end{proof}











\section{Conclusion}\label{sec:conclusion}
 In this paper, we derive lower bounds on the sub-packetization of optimal-access MSR codes in the scenario of multiple-node failures. 
 Besides, we offer a more detailed analysis by employing generating functions. 
 Below, we present a comparison of these bounds.
 By setting $c=h/2$ in \cref{generating_1} and \cref{generating_2}, we can compute that
\[
\alpha = \frac{-\frac{h}{2} + \sqrt{\frac{h^2}{4} + 2h(d-k+h)}}{h},
\]
which leads to the bound
\[
\ell \geq \Theta\left(\alpha^{\frac{2n}{d-k+h}}\right).
\]
On the other hand, the bounds in \cref{thm_constant2} and \cref{thm_nonconstant2} give
\[
\ell \geq \Theta\left(\left(\frac{d-k+h}{h}\right)^{\frac{n}{d-k+h}}\right).
\]
It is not hard to check that  \(\alpha^2 > \frac{d-k+h}{h}\).
Hence, for high-rate MSR codes and sufficiently large \(n\), the bounds in Theorems~\ref{generating_1} and~\ref{generating_2} are greater than those in Theorems~\ref{thm_constant2} and~\ref{thm_nonconstant2}. Therefore, unlike the case when \(h = 1\), the generalized bounds in Theorems~\ref{thm_constant2} and~\ref{thm_nonconstant2} are generally not tight.

\appendices
\section{The Cut-Set Bound on Centralized MSR Codes With Code Rate No Less Than 1/2}\label{proof_cutset}
\begin{proof}
(The proof method is proposed by Alrabiah and Guruswami in \cite{Guruswami2019})
    Let \(\C\) be an \((n,k,\ell)\) MDS array code over \(\F_q\). 
    We write a codeword of \(\C\) as \(C_1,\ldots,C_n\). 
    Give two integers \( h,d\) such that \(2\le h\le n-k\) and \(k+1\le d\le n-h\).  
    For any set of \(h\) failed nodes \(\mathcal{F} \subseteq [n]\) and any set of \(d\) helper nodes \(\mathcal{R} \subseteq [n] \setminus \mathcal{F}\), the code \(\mathcal{C}\) will repair \(\mathcal{F}\) by having each \(C_j\) for \(j \in \mathcal{R}\) transmit \(\beta_{j, \mathcal{F}}\) bits to the single data center.
   Since code rate \(\frac{k}{n}\ge\frac{1}{2}\), \(k-h\ge k-(n-k)\ge 0\). By  MDS property, any subset \( M \subseteq \mathcal{R} \) of size \( k - h \), together with \( h \) failed nodes in \( \mathcal{F} \), can reconstruct the entire codeword. Thus
\begin{equation}
    \sum\limits_{i\in M}\ell+\sum\limits_{j\in\cR\bl M} \beta_{j,\cF} \ge k\ell.
\end{equation}
Namely,
\begin{equation}\label{single_sum}
    \sum\limits_{j\in\cR\bl M} \beta_{j,\cF} \ge h\ell.
\end{equation}
By summing over all possible \(M\), we get
\begin{align}
    \sum\limits_{\substack{M\subseteq \cR,\\|M|=k-h}}\sum\limits_{j\in\cR\bl M} \beta_{j,\cF} \ge \sum\limits_{\substack{M\subseteq \cR,\\|M|=k-h}}h\ell,\\
     \sum\limits_{j\in\cR} \beta_{j,\cF} \ge \frac{dhl}{d-k+h}\label{multisum}.
\end{align}
Furthermore, to achieve equality, it is necessary for equality to hold in (\ref{single_sum}) for all \(M\). It is not difficult to verify that \(\beta_{j, \mathcal{F}} = \frac{hl}{d - k + h}\) for all \(j \in \mathcal{R}\).
\end{proof}

\bibliographystyle{ieeetr}
\bibliography{ref}
\end{document}